\documentclass[conference]{IEEEtran}
\IEEEoverridecommandlockouts

\usepackage{cite}
\usepackage{amsmath,amssymb,amsfonts,amsthm}
\usepackage{graphicx}
\usepackage{textcomp}
\usepackage{xcolor}
\usepackage{multirow}
\usepackage{booktabs}
\usepackage{float}
\usepackage{placeins}
\usepackage{enumitem}
\usepackage{algorithm}
\usepackage{algpseudocode}
\usepackage{bbm}
\usepackage{mathrsfs}

\makeatletter
\newcommand\fs@ruledtopgap{%
  \def\@fs@cfont{\bfseries}%
  \let\@fs@capt\floatc@ruled
  \def\@fs@pre{\vskip 6pt\hrule height.8pt depth0pt \kern2pt}%
  \def\@fs@post{\kern2pt\hrule\relax}%
  \def\@fs@mid{\kern2pt\hrule\kern2pt}%
  \let\@fs@iftopcapt\iftrue
}
\makeatother

\floatstyle{ruledtopgap}
\restylefloat{algorithm}

\newtheorem{theorem}{Theorem}
\newtheorem{lemma}[theorem]{Lemma}
\newtheorem{proposition}[theorem]{Proposition}

\makeatletter
\g@addto@macro\normalsize{%
  \setlength{\abovedisplayskip}{4pt plus 1pt minus 1pt}%
  \setlength{\belowdisplayskip}{4pt plus 1pt minus 1pt}%
  \setlength{\abovedisplayshortskip}{2pt plus 1pt minus 1pt}%
  \setlength{\belowdisplayshortskip}{2pt plus 1pt minus 1pt}%
}
\makeatother

\def\BibTeX{{\rm B\kern-.05em{\sc i\kern-.025em b}\kern-.08em
    T\kern-.1667em\lower.7ex\hbox{E}\kern-.125emX}}

\makeatother

\begin{document}

\title{Reliability-Aware Scheduling for\\
Digital Twin Maintenance}

\author{
\IEEEauthorblockN{Milica~Jankov, Carlo~Fischione}
\IEEEauthorblockA{School of Electrical Engineering and Computer Science,
KTH Royal Institute of Technology, Stockholm, Sweden\\
Email: \{milicaj, carlofi\}@kth.se
}
\thanks{This project has received funding from the European Union's Horizon Europe programme under grant agreement No.~101137954. The authors would like to thank the BATTwin consortium for supporting this research.}
}

\maketitle

\begin{abstract}
In Industrial Internet of Things systems, learning-enabled Digital Twins
(DTs) support remote monitoring by using data reported by distributed devices
to maintain digital representations of physical processes. When uplink
resources are limited, a base station cannot collect new observations from
every device at every communication slot and must decide which devices should transmit. This
decision becomes challenging when the physical process changes after the DT
models have already been trained. In such cases, recent observations alone may not keep
the DT accurate, because the learned model may no longer match the underlying
process.
This paper studies how to schedule observation requests so that the DT
maintained at the base station remains close to the true physical process. We
define the Ensemble Disagreement Indicator (EDI) as an uncertainty measure computed at
the base station from the spread among estimates produced by independently
trained DT predictors. Building on EDI, we propose R-VoU, a reliability-aware
value-of-update scheduler that prioritizes the observations expected to most
improve the DT by reducing a cost based on uncertainty and predicted DT error. The DT is further adjusted online using the
difference between received observations and current predictions. Experiments on process manufacturing data show that, under limited
communication budgets, R-VoU achieves the lowest combined cost and DT
estimation error among the compared schedulers that use only information
available before each scheduling decision.
\end{abstract}

\vspace{-1.70mm}
\section{Introduction}

Industrial Internet of Things (IIoT) systems increasingly rely on learning-enabled
Digital Twins (DTs) for remote monitoring, anomaly detection, and control. The setting
considered here involves learning-enabled receiver-side DTs maintained at a
base station (BS). Each DT uses intermittently received measurements to
estimate and track the state of an \mbox{evolving} physical process~\cite{tao19,pappas2021semantics}. In resource-constrained uplink networks,
the BS cannot receive fresh measurements from all devices continuously. We
therefore consider a pull-based update scheme, in which the BS requests
observations from selected devices under a per-slot communication budget~\cite{chiariotti2022qaoi, agheli2024eff}. The resulting scheduling question is which devices should be pulled so that
the BS-maintained DTs remain \textit{reliable}, i.e., accurate and synchronized with
the physical process \cite{zhang2024knowledge}.

This question becomes more difficult under \emph{concept drift}, which refers
to deployment-time changes in operating conditions or process regimes that make
models trained on historical data less aligned with the current physical
process~\cite{Lu_18conceptdrift}. Under concept drift, temporal freshness alone is not
enough to guarantee DT reliability. Even a recently updated DT may be inaccurate
if its predictor is mismatched to the current operating regime. This issue is
especially relevant in industrial process manufacturing, where limited communication
\mbox{opportunities} must support accurate state tracking under changing process
conditions~\cite{tao19}.

Existing freshness metrics and related scheduling methods provide useful baselines for reliability, but
they do not fully capture the failure mode induced by concept drift. Age of Information, Age of
Incorrect Information, Query Age of Information, and Value of
Information methods rank updates according to elapsed time, mismatch, query
relevance, or expected utility
\cite{sun2022age,maataouk2020aoii,chiariotti2022qaoi,Molin2019VoI}.
However, under concept drift, reducing receiver-side uncertainty and reducing
receiver-side DT error are related but not identical reliability objectives.
To clarify the reliability objective, we distinguish between receiver-side
epistemic uncertainty and receiver-side DT error. 

Epistemic uncertainty refers
to uncertainty about the learned DT model itself~\cite{waegeman21}. It arises
when the current prediction is not well supported by the available training
data, for example when concept drift moves the process to an operating regime
that was weakly represented during training. To capture this uncertainty, we define the \emph{Ensemble Disagreement Indicator} (EDI) from the disagreement among independently trained ensemble predictors, following the deep ensemble approach~\cite{nips2017}.We will show that EDI is observable at the BS because it is computed solely from receiver-side DT states and does not require the latent physical state.

Receiver-side DT error is defined separately as the mismatch between the DT
estimate and the true physical state, and thus directly quantifies tracking accuracy.
Because the true physical state monitored by the devices is not available to the scheduler at runtime,
this error is latent and must be predicted from BS-observable information.
Thus, EDI provides a BS-observable uncertainty signal, whereas DT error is
the latent reliability target.

The main contribution of this paper is the design of \mbox{R-VoU}, a novel
reliability-aware uplink pull scheduler for receiver-side DT maintenance under
concept drift. Under a one-step predictive formulation, \mbox{R-VoU} allocates pull opportunities to the devices expected to provide the largest positive reductions in DT
reliability cost. The second contribution formulates the
reliability model by defining EDI as a BS-observable uncertainty
signal from DT ensemble disagreement, distinguishing it from latent DT error,
and using pull residuals for online correction. The third contribution investigates the design and analysis of the learned scheduling rule. We train action-conditioned predictors for next-slot EDI and
DT error risk, derive a closed-form rule for selecting up to \(K\) devices,
establish a selection-loss bound, and validate the method on process
manufacturing data.
\vspace{-2.27mm}

\section{System Model and Pull-Scheduling Problem}
\label{sec:maintenance_model}

We consider a centralized time-slotted pull-maintenance system with one
base station (BS) and \(N\) devices. This system model follows pull-based
freshness and status-update models, where a receiver or controller requests
updates from selected devices under limited communication opportunities~\cite{chiariotti2022qaoi,kriole23,agheli2024eff}. The BS maintains one receiver-side DT
per device and can successfully receive fresh observations from at most
\(K\) devices per slot. In each slot, the BS schedules the devices to pull;
devices that are not scheduled do not transmit. We abstract
communication through successful receptions and focus on allocating the
limited pull opportunities so that the receiver-side DTs remain reliable
under concept drift.

Fig.~\ref{fig:system_model} summarizes the architecture. Pulled observations
refresh the selected DTs and reveal prediction residuals, defined as the
differences between received observations and current DT predictions. These
residuals update a shared online correction module that helps the
receiver-side DTs adapt under concept drift.
\vspace{-2mm}

\begin{figure}[!t]
    \centering
    \includegraphics[width=0.99
    \columnwidth]{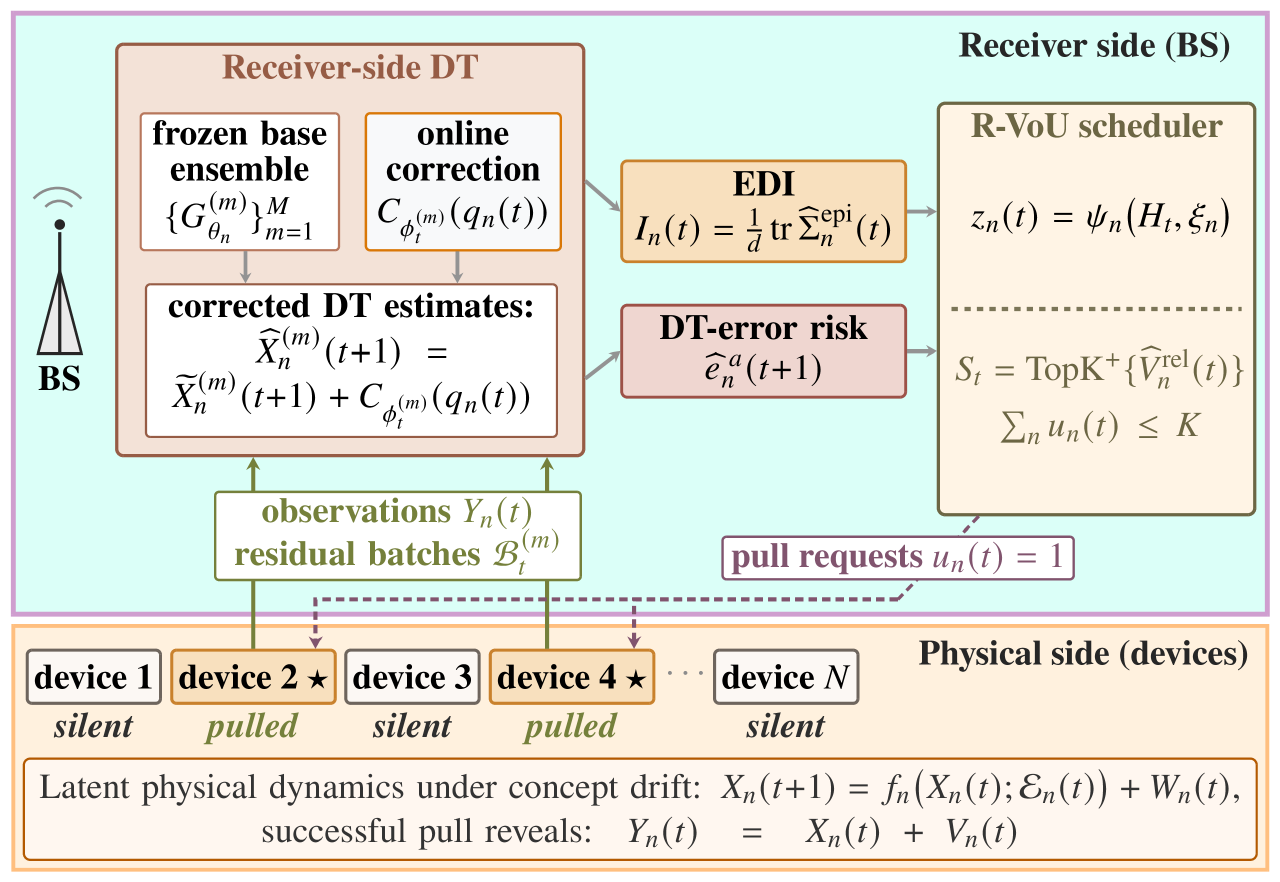}
    \vspace{-2mm}
    \caption{Pull-maintenance architecture. At each slot, the BS pulls up to \(K\) devices. The received observations refresh the selected receiver-side DTs and provide prediction residuals for online correction.}
    \vspace{-1.7mm}
    \label{fig:system_model}
\end{figure}

\vspace{-0.25mm}
\subsection{Receiver-Side DT with Online Correction}

We consider a time-slotted system indexed by
\(t\in\mathbb Z_{\ge 0}\). The device set is
\(\mathcal N=\{1,\ldots,N\}\). For each device \(n\in\mathcal N\),
let \(X_n(t)\in\mathbb R^d\) denote the latent physical state at
slot \(t\), where \(d\) is the
dimension of the per-device state. This state is not directly available at the BS. The BS
therefore maintains a receiver-side DT estimate and receives
information about \(X_n(t)\) only through noisy observations
obtained from successful pulls.

We model each device as evolving under a latent operating regime
\(\mathcal E_n(t)\). Changes in \(\mathcal E_n(t)\) correspond to the
concept drift setting introduced in Section~I. The state evolves according to
\begin{equation}
\label{eq:state_dynamics}
X_n(t+1)
=
f_n\!\big(X_n(t);\mathcal E_n(t)\big)+W_n(t),
\end{equation}
where \(f_n(\cdot;\mathcal E_n(t)):\mathbb R^d\to\mathbb R^d\) is the
state transition map of device \(n\) under regime \(\mathcal E_n(t)\). The map is not restricted to be
linear, and \(W_n(t)\in\mathbb R^d\) is zero-mean process noise.

When device \(n\) is successfully pulled at slot \(t\), the BS receives the
observation
\begin{equation}
\label{eq:observation_model}
Y_n(t)
=
X_n(t)+V_n(t),
\end{equation}
where \(V_n(t)\) is zero-mean measurement noise. We denote the successful
pull event by \(u_n(t)=1\). If \(u_n(t)=0\), then no fresh observation from
device \(n\) is received at slot \(t\).

For each device \(n\), the BS stores \(M\) pretrained base predictors
\(\{G_{\theta_n^{(m)}}:\mathbb R^d\to\mathbb R^d\}_{m=1}^{M}\), where
\(m\) indexes the ensemble member. The parameters
\(\theta_n^{(m)}\) are trained offline and kept fixed during deployment.
Since the ensemble members are trained independently, their disagreement
can be used as a BS-side measure of epistemic uncertainty, i.e., uncertainty
about the learned DT model~\cite{waegeman21}.

To adapt the frozen base predictors under concept drift, we introduce an online
correction module based on recursive least squares (RLS)~\cite{RLS}, which sequentially
estimates the linear correction parameters from received residual samples. Let
\(q_n(t)\in\mathbb R^p\) denote a BS-observable adaptation context for device
\(n\). For ensemble member \(m\), the correction is
\begin{equation}
\label{eq:rls_corrector}
C_{\phi_t^{(m)}}(q_n(t))
=
{W_t^{(m)}}^\top q_n(t),
\qquad
\phi_t^{(m)}=(W_t^{(m)},P_t^{(m)}),
\end{equation}
where \(W_t^{(m)}\in\mathbb R^{p\times d}\) is the RLS coefficient matrix
and \(P_t^{(m)}\) is the RLS covariance matrix. The coefficient matrix
\(W_t^{(m)}\) is shared across devices for the same ensemble member \(m\).
The correction remains device-dependent through the context \(q_n(t)\).

Let \(\tilde X_n^{(m)}(t+1)\) denote the base prediction before applying the
online correction, and let \(\hat X_n^{(m)}(t+1)\) denote the corrected
receiver-side DT estimate. The base prediction is updated as
\begin{equation}
\label{eq:base_dt_recursion_joint}
\tilde X_n^{(m)}(t+1)
=
\begin{cases}
G_{\theta_n^{(m)}}\!\big(Y_n(t)\big), & u_n(t)=1,\\
G_{\theta_n^{(m)}}\!\big(\hat X_n^{(m)}(t)\big), & u_n(t)=0.
\end{cases}
\end{equation}
Thus, the base prediction uses the newly received observation when
\(u_n(t)=1\) and the previous corrected DT estimate when \(u_n(t)=0\). The corrected estimate is then given by
\begin{equation}
\label{eq:corrected_dt_recursion_joint}
\hat X_n^{(m)}(t+1)
=
\tilde X_n^{(m)}(t+1)
+
C_{\phi_t^{(m)}}\!\big(q_n(t)\big).
\end{equation}
A pulled observation also provides information for online adaptation. Let
\(\tilde X_n^{(m)}(t)\) denote stored pre-correction base estimate for the current slot, obtained by the recursion in \eqref{eq:base_dt_recursion_joint} at the previous update. When device \(n\) is pulled, we define
the base-residual target
\begin{equation}
\label{eq:base_residual_supervision}
b_n^{(m)}(t)
\triangleq
Y_n(t)-\tilde X_n^{(m)}(t).
\end{equation}
This residual is the additive correction that would make the current base
prediction match the received observation.

For each ensemble member \(m\), the adaptation batch at slot \(t\) contains
one pair \((q_n(t),b_n^{(m)}(t))\) for every successfully pulled device:
\vspace{-2.3mm}
\begin{equation}
\label{eq:adaptation_batch}
\mathcal B_t^{(m)}
\triangleq
\big\{(q_n(t),b_n^{(m)}(t)):u_n(t)=1\big\}.
\end{equation}
The RLS state is updated according to
\vspace{-0.4mm}
\begin{equation}
\label{eq:adaptation_update}
\phi_{t+1}^{(m)}
=
\mathcal A_{\mathrm{RLS}}^{(m)}
\big(\phi_t^{(m)},\mathcal B_t^{(m)}\big),
\end{equation}
where \(\mathcal A_{\mathrm{RLS}}^{(m)}\)  denotes the standard RLS update
operator applied to the pulled batch \cite{RLS}.
If no device is pulled at slot \(t\), then
\(\mathcal B_t^{(m)}=\varnothing\) and \(\phi_{t+1}^{(m)}=\phi_t^{(m)}\).
Because the RLS coefficient matrix \(W_t^{(m)}\) is shared across devices
for each ensemble member \(m\), a base-residual target obtained from one
pulled device can update this shared coefficient matrix and thereby affect
future corrections for other devices with related contexts.

The ensemble-mean receiver-side DT estimate is
\vspace{-1mm}
\begin{equation}
\label{eq:ensemble_mean}
\bar X_n(t)
\triangleq
\frac{1}{M}\sum_{m=1}^{M}\hat X_n^{(m)}(t).
\end{equation}
We define the Ensemble Disagreement Indicator (EDI) as the average per-component sample variance of the corrected DT ensemble,
\vspace{-1.6mm}
\begin{equation}
\label{eq:aoint_new}
I_n(t)
\triangleq
\frac{1}{d(M-1)}
\sum_{m=1}^{M}
\left\|
\hat X_n^{(m)}(t)-\bar X_n(t)
\right\|_2^2 .
\end{equation}
Because \(I_n(t)\) depends only on receiver-side DT states, it is
observable at the BS and provides a causal uncertainty signal. Larger values
indicate stronger ensemble disagreement and hence greater epistemic
uncertainty.  Under model mismatch, EDI may increase during passive propagation
and decrease after informative receptions.

Tracking accuracy is described by a separate metric. We define the
receiver-side DT error as the mismatch between the ensemble-mean estimate
and the latent physical state,
\vspace{-1.2mm}
\begin{equation}
\label{eq:dt_error_system}
e_n(t)
\triangleq
\frac{1}{d}
\left\|
\bar X_n(t)-X_n(t)
\right\|_2 .
\end{equation}
This error depends on \(X_n(t)\), which is not available to the BS during
runtime. Hence, \(e_n(t)\) is not directly observable by the scheduler. In
offline experiments with recorded trajectories, the reference physical state
is used to compute \(e_n(t)\) for training labels and evaluation. Thus,
EDI provides a BS-observable uncertainty signal, whereas DT error is the
latent reliability target that must be inferred from the information available at the BS.
\vspace{-1.10mm}
\subsection{Reliability-Aware Pull-Scheduling Objective}
\vspace{-0.2mm}
At the beginning of each slot \(t\), the BS decides which devices to pull
before receiving any new observations. Let
\(
u(t) \triangleq (u_1(t),\ldots,u_N(t))\in\{0,1\}^N
\)
denote the slot-\(t\) pull-decision vector. Under the successful-reception
abstraction, \(u_n(t)=1\) means that device \(n\) is pulled and its
observation is received by the BS during slot \(t\). If \(u_n(t)=0\), no
fresh observation from device \(n\) is received in that slot.

The BS can receive observations from at most \(K\) devices per slot, where
\(K\in\{0,\ldots,N\}\). Hence, the feasible action set is
\vspace{-2mm}
\begin{equation}
\label{eq:feasible_action_set}
\mathcal U_K
\triangleq
\left\{
u\in\{0,1\}^{N}:
\sum_{n=1}^{N}u_n\le K
\right\}.
\end{equation}
We also define the set of pulled devices as
\vspace{-1.5mm}
\[
\mathcal S_t \triangleq \{n\in\mathcal N:u_n(t)=1\}.
\]
\noindent After the BS selects \(\mathcal S_t\), it receives
\(\{Y_n(t):n\in\mathcal S_t\}\). These observations refresh the selected
receiver-side DT states through \eqref{eq:base_dt_recursion_joint} and
\eqref{eq:corrected_dt_recursion_joint}. They also provide residual targets
for the RLS adaptation batches in \eqref{eq:adaptation_batch}. Devices that
are not pulled evolve through the skip branch of
\eqref{eq:base_dt_recursion_joint} and provide no new residual information
in that slot.

The scheduling decision is causal. Thus, \(u(t)\) must be chosen using only information available at the BS
before any observations are received during slot \(t\). Let
\(I(t)\triangleq(I_1(t),\ldots,I_N(t))\) collect the current EDI values across devices. Let
\(
\mathcal M_t \triangleq \{\phi_t^{(m)}\}_{m=1}^{M}
\)
denote the current RLS correction state, and let
\(\mathcal D_t^{\mathrm{res}}\) denote the residual information obtained
from previous pulls. The dynamic scheduler information is
\vspace{-1.2mm}
\begin{equation}
\label{eq:scheduler_information_reliability}
H_t
\triangleq
\Big(
\{I(\tau)\}_{\tau\le t},
\{u(\tau)\}_{\tau<t},
\mathcal D_t^{\mathrm{res}},
\mathcal M_t
\Big).
\end{equation}
\vspace{-3.9mm}

\noindent The receiver-side DT error vector
\(
e(t)\triangleq(e_1(t),\ldots,e_N(t))
\)
is not included in \(H_t\), because it depends on the latent physical states
and is not available to the runtime scheduler.

The scheduler is also provided with fixed device metadata. For each device
\(n\), let \(\xi_n\in\mathbb R^{r}\) denote metadata that is known before
scheduling and does not change over time, and let \(\omega_n>0\) denote the
reliability priority weight of the device. A causal pull-scheduling policy is
a sequence \(\pi=\{\pi_t\}_{t\ge0}\) such that
\[
\pi_t:
\big(
H_t,
\{\xi_n\}_{n=1}^{N},
\{\omega_n\}_{n=1}^{N}
\big)
\mapsto
u(t)\in\mathcal U_K .
\]
Thus, the policy uses the history of EDI values, past pull decisions, residual history,
the current RLS correction state, and fixed device metadata. It cannot use
the latent DT error vector \(e(t)\).

Since EDI and DT error have different scales, we normalize them using
training-set constants \(s_I>0\) and \(s_e>0\), fixed during evaluation. For a design parameter \(\alpha\in[0,1]\), the
per-device reliability cost is
\vspace{-4mm}

\begin{equation}
\label{eq:reliability_cost_system}
\mathcal J_n(t)
\triangleq
\omega_n
\left[
\alpha \frac{I_n(t)}{s_I}
+
(1-\alpha)\frac{e_n(t)}{s_e}
\right].
\end{equation}
The parameter \(\alpha\) balances EDI and DT error:
\(\alpha=1\) yields an EDI-only objective, while \(\alpha=0\) gives an objective based only on DT error.

We propose the following reliability-aware pull-scheduling problem:
\vspace{-4mm}
\begin{equation}
\label{eq:ideal_objective_reliability}
\min_{\pi}
\limsup_{T\to\infty}
\frac{1}{T}
\sum_{t=1}^{T}
\mathbb E_{\pi}
\left[
\sum_{n=1}^{N}\mathcal J_n(t)
\right],
\end{equation}
where the expectation is taken over the process noise, measurement noise,
latent regime evolution, and any randomization in the policy.

Problem~\eqref{eq:ideal_objective_reliability} provides an ideal objective,
but it does not directly yield a causal scheduling policy for the BS. Its direct
solution would require modeling the controlled evolution of \(I_n(t)\),
\(e_n(t)\), and \(\mathcal M_t\) under pull decisions, skipped updates,
online RLS adaptation, and concept drift. Moreover, \(e_n(t)\) is not
observable by the BS during deployment. We therefore cannot solve the ideal
problem directly. Instead, Section~\ref{sec:scheduling} investigates a learned
one-step prediction model that estimates, from BS-observable information, the
next-slot EDI and DT error risk under skip and pull actions. These estimates
are then used to compute the predicted reliability-cost reduction of each pull
decision.
\vspace{-1.02mm}
\section{Reliability-Aware Predictive Pull Scheduling}
\label{sec:scheduling}
\vspace{-0.2mm}
In this section, we introduce R-VoU, a reliability-aware value-of-update
scheduler based on learned one-step predictions. Offline, R-VoU learns
action-conditioned predictors for next-slot EDI and DT error risk from
BS-observable features, whereas online, the BS evaluates each device under
the skip and pull actions using the current context. The predictions are
combined into a reliability cost, and the BS selects up to \(K\) devices with
the largest positive predicted cost reductions.
\vspace{-0.9mm}
\subsection{Prediction Under Skip and Pull Actions}
\vspace{-0.50mm}
\noindent
For each device \(n\), let
\vspace{-1.30mm}
\begin{equation}
\label{eq:rvou_feature_vector}
z_n(t)=\psi_n(H_t,\xi_n)
\end{equation}
denote the scheduling feature vector available before the slot-\(t\)
pull decision. It is constructed from BS-observable information \(H_t\) and
the fixed device metadata \(\xi_n\). The latent DT error \(e_n(t)\) is not included.

To train the action-conditioned predictors, we use a local hypothetical action
index \(a\in\{0,1\}\), distinct from the realized scheduling variable
\(u_n(t)\). The index \(a\) defines the skip \((a=0)\) and pull \((a=1)\)
branches used to construct prediction targets, while \(u_n(t)\) denotes the
actual pull decision made by the BS at slot \(t\). For each training pair \((n,t)\), we define the next-slot prediction target as
\vspace{-0.1mm}
\begin{equation}
\label{eq:rvou_targets}
y_{n,t}^{(a)}
\triangleq
\begin{bmatrix}
y_{I,n,t}^{(a)}\\
y_{e,n,t}^{(a)}
\end{bmatrix}
=
\begin{bmatrix}
I_n^a(t+1)\\
e_n^a(t+1)
\end{bmatrix},
\qquad a\in\{0,1\}.
\end{equation}
The superscript \(a\) denotes the next-slot outcome obtained under the
corresponding hypothetical branch.  During training only, the recorded
trajectories are used to construct targets for both actions, skip and pull,
using the same DT recursion and, when needed, the same RLS update as at
runtime. In the pull branch, the received observation is used to form the
base-residual target \(Y_n(t)-\tilde X_n^{(m)}(t)\). In the skip branch, no new
residual target is formed for device \(n\).

At runtime, the scheduler does not observe the targets in
\eqref{eq:rvou_targets}. It only uses predictors trained from them.
\vspace{-1.9mm}
\subsection{Shared Linear Prediction Heads}
\vspace{-0.4mm}
\noindent
R-VoU uses shared prediction heads across devices. Let
\vspace{-0.8mm}
\[
\tilde z_n(t)=[1,z_n(t)^\top]^\top
\]
be the feature vector augmented with an intercept term. For each action \(a\in\{0,1\}\), the predicted EDI and DT error risk
at slot \(t+1\) are
\vspace{-1.95mm}
\begin{equation}
\label{eq:rvou_linear_heads}
\widehat y_n^a(t+1)
\triangleq
\begin{bmatrix}
\widehat I_n^a(t+1)\\
\widehat e_n^a(t+1)
\end{bmatrix}
=
B_a^\top \tilde z_n(t).
\end{equation}
Here, \(B_a\) is the coefficient matrix of the prediction head for action
\(a\). Its two columns are denoted by \(\beta_{I,a}\) and \(\beta_{e,a}\),
which map the feature vector \(\tilde z_n(t)\) to the EDI prediction and
the DT error risk prediction, respectively.

Let \(\mathcal T_{\mathrm{tr}}\) denote the set of device-slot pairs used for
training. For each action \(a\in\{0,1\}\), where \(a=0\) denotes skip and
\(a=1\) denotes pull, the corresponding prediction head is learned by
weighted multi-output ridge regression \cite{Ridge}. The output-weighting
matrix is \(\Gamma=\mathrm{diag}(1,\mu_e)\), where \(\mu_e\) controls the
relative weight of the DT error target:
\vspace{-0.7mm}
\begin{equation}
\label{eq:rvou_ridge}
\begin{aligned}
\widehat B_a \in \arg\min_{B_a}\quad
& \sum_{(n,t)\in\mathcal T_{\mathrm{tr}}}
\left\|
\Gamma^{1/2}
\big(
B_a^\top \tilde z_n(t)-y_{n,t}^{(a)}
\big)
\right\|_2^2  \\
\vspace{-0.8mm}
&\hspace{2.3cm} + \lambda_{\mathrm{reg}}\|B_a\|_F^2 .
\end{aligned}
\end{equation}
Here, \(\lambda_{\mathrm{reg}}\) controls the \(\ell_2\) regularization.
All features are standardized using training-set statistics. Since EDI and
DT error are nonnegative quantities, negative predictions are set to zero
before computing reliability costs. The priority weight \(\omega_n\) is not
used by the prediction heads and enters only through the reliability cost.
\vspace{-1.5mm}
\subsection{Reliability-Aware Value of Update}

The action-conditioned predictions are converted into a predicted
reliability cost using the same normalization as in
\eqref{eq:reliability_cost_system}. For device \(n\) and action
\(a\in\{0,1\}\),
\begin{equation}
\label{eq:predicted_reliability_cost}
\widehat{\mathcal J}_n^a(t+1)
\triangleq
\omega_n
\left[
\alpha\frac{\widehat I_n^a(t+1)}{s_I}
+
(1-\alpha)\frac{\widehat e_n^a(t+1)}{s_e}
\right].
\end{equation}
Here, \(a=0\) denotes skip and \(a=1\) pull, while
\(\alpha\in[0,1]\) balances predicted uncertainty and DT error risk.

The reliability-aware value of update is defined as the predicted
cost reduction obtained by pulling the device rather than skipping it:
\begin{equation}
\label{eq:rvou_score}
\widehat V_n^{\mathrm{rel}}(t)
\triangleq
\widehat{\mathcal J}_n^0(t+1)
-
\widehat{\mathcal J}_n^1(t+1).
\end{equation}
Thus, a positive value indicates that pulling device \(n\) is predicted
to reduce the next-slot reliability cost. The uncertainty-only predictive ablation is obtained by setting \(\alpha=1\);
we denote this ablation by EDI-VoU.
\vspace{-1.55mm}
\subsection{Top-\(K\)-Positive Rule and Selection-Loss Bound}
\vspace{-0.5mm}
Given the predicted costs, R-VoU solves the following one-step predictive
scheduling problem:
\vspace{-0.1mm}
\begin{equation}
\label{eq:rvou_proxy_problem}
\min_{u(t)\in\mathcal U_K}
\sum_{n=1}^{N}
\left[
(1-u_n(t))\widehat{\mathcal J}_n^0(t+1)
+
u_n(t)\widehat{\mathcal J}_n^1(t+1)
\right].
\end{equation}
For the ranking step, we define the reliability score of device \(n\) at slot
\(t\) as
\vspace{-2mm}
\[
s_n(t)\triangleq \widehat V_n^{\mathrm{rel}}(t).
\]
Thus, \(s_n(t)\) is the predicted reduction in reliability cost obtained by
pulling device \(n\) rather than skipping it. For any vector
\(x=(x_1,\ldots,x_N)\), we define \(\operatorname{TopK}_{+}(x)\) as the set of
indices corresponding to the largest strictly positive entries of \(x\), up to
the budget \(K\). Equivalently,
\vspace{-1mm}
\[
\operatorname{TopK}_{+}(x)
\in
\arg\max_{S\subseteq\mathcal N,\, |S|\le K}
\sum_{n\in S}x_n,
\]
with the convention that entries with \(x_n\le 0\) are not selected. Therefore, \(\operatorname{TopK}_{+}\) returns \(K\)
devices only when at least \(K\) entries of \(x\) are positive.
\vspace{-2.1mm}
\begin{lemma}[Top-\(K\) positive selection rule]
\label{lem:rvou_topk}
Let
\[
\widehat S_t=\operatorname{TopK}_{+}(s(t)),
\qquad
s(t)=(s_1(t),\ldots,s_N(t)).
\]
Define the binary decision vector \(u^\star(t)\) by
\vspace{-1mm}
\[
u_n^\star(t)=\mathbf 1\{n\in \widehat S_t\},
\qquad n=1,\ldots,N.
\]
Then \(u^\star(t)\) is an optimal solution of \eqref{eq:rvou_proxy_problem}. The pull set induced
by this decision is \(\widehat S_t\). Hence, the scheduler pulls up to \(K\)
devices and pulls exactly \(K\) devices only when at least \(K\) devices have
positive reliability scores.
\end{lemma}
\vspace{-1mm}
\noindent\emph{Proof sketch.}
Substituting \eqref{eq:rvou_score} into \eqref{eq:rvou_proxy_problem}
removes terms independent of \(u(t)\), so the problem becomes maximizing
\(\sum_{n=1}^{N}u_n(t)s_n(t)\) subject to \(\sum_n u_n(t)\le K\).
Thus, the rule sets \(u_n(t)=1\) for the devices with the largest positive
scores, up to the budget \(K\), and sets the remaining entries to zero.
\(\square\)

Because \(s_n(t)\) is computed from the predicted reduction in the composite
reliability cost, it generalizes uncertainty-only scheduling. Two devices with
similar predicted EDI reduction can be ranked differently if their predicted
DT error risk reductions differ.

We next give a current-slot guarantee that connects selection quality to
reliability-score prediction accuracy. Let \(\mathcal V_n^\star(t)\) be a reference reliability value for device
\(n\), and define
\(\mathcal V^\star(t)\triangleq(\mathcal V_1^\star(t),\ldots,
\mathcal V_N^\star(t))\). Using
\(\operatorname{TopK}_{+}\), we define two selected sets: \(S_t^\star\) is
the set selected by the reference reliability values, and \(\widehat S_t\)
is the set selected by \mbox{R-VoU} using the learned reliability scores:
\vspace{-1.6mm}
\[
S_t^\star
=
\operatorname{TopK}_{+}(\mathcal V^\star(t)),
\qquad
\widehat S_t
=
\operatorname{TopK}_{+}(s(t)).
\]
\noindent The selection loss incurred by using the learned scores is
\begin{equation}
\label{eq:rvou_selection_loss}
\mathcal L_t^{\mathrm{sel},\star}
\triangleq
\sum_{n\in S_t^\star}\mathcal V_n^\star(t)
-
\sum_{n\in\widehat S_t}\mathcal V_n^\star(t).
\end{equation}
\vspace{-4.9mm}
\begin{proposition}[Prediction-error selection-loss bound]
\label{prop:rvou_selection_loss}
Suppose that the learned reliability scores satisfy
\vspace{-1mm}
\[
\left|
s_n(t)-\mathcal V_n^\star(t)
\right|
\le
\varepsilon_n(t),
\qquad n=1,\ldots,N .
\]
Then
\begin{equation}
\label{eq:rvou_selection_loss_bound}
0
\le
\mathcal L_t^{\mathrm{sel},\star}
\le
\sum_{n\in S_t^\star}\varepsilon_n(t)
+
\sum_{n\in\widehat S_t}\varepsilon_n(t)
\le
2K\varepsilon_{\max}(t),
\end{equation}
where \(\varepsilon_{\max}(t)=\max_n\varepsilon_n(t)\).
\end{proposition}
\vspace{-2mm}
\noindent\emph{Proof sketch.}
The nonnegativity follows because \(S_t^\star\) maximizes the reference
score. Adding and subtracting the learned scores in
\eqref{eq:rvou_selection_loss}, and using
\(\sum_{n\in\widehat S_t}s_n(t)\ge
\sum_{n\in S_t^\star}s_n(t)\), which follows from
\(\widehat S_t=\operatorname{TopK}_{+}(s(t))\), leaves only the prediction
errors on \(S_t^\star\) and \(\widehat S_t\). Bounding these errors by
\(\varepsilon_n(t)\) gives the first inequality in
\eqref{eq:rvou_selection_loss_bound}; the last inequality follows from
\(|S_t^\star|,|\widehat S_t|\le K\). \(\square\)

Proposition~\ref{prop:rvou_selection_loss} shows that the quality of the
current-slot selection depends on how accurately the learned reliability scores
approximate the reference values: if these score errors are small, then the
loss relative to the reference selection is also small and is bounded by
\(2K\varepsilon_{\max}(t)\).
A useful choice for the reference value is the true one-step reliability value
\vspace{-0.1mm}
\begin{equation}
\label{eq:true_reference_value}
\mathcal V_n^\star(t)
=
\mathcal J_n^{0,\star}(t+1)
-
\mathcal J_n^{1,\star}(t+1),
\end{equation}
where \(\mathcal J_n^{a,\star}(t+1)\) is the reliability cost under the
reference next-slot outcome for action \(a\).
\vspace{-1.80mm}
\subsection{Runtime Policy}
\vspace{-1mm}
At runtime, the prediction heads are fixed, while the RLS correction state
\(\mathcal M_t=\{\phi_t^{(m)}\}_{m=1}^{M}\) continues to evolve from the
observations received through successful pulls.
Algorithm~\ref{alg:rvou_runtime} summarizes the causal R-VoU policy.

\begin{algorithm}[t]
\caption{Runtime R-VoU at Slot \(t\)}
\label{alg:rvou_runtime}
\small
\begin{algorithmic}[1]

\Require \(H_t,\mathcal M_t,K,\alpha,s_I,s_e,
\{\omega_n,\xi_n\}_{n=1}^{N}\), learned heads \(\widehat B\)
\For{\(n=1,\ldots,N\)}
    \State \(z_n(t)\gets\psi_n(H_t,\xi_n)\)
    \State Compute \(\widehat{\mathcal J}_n^0(t+1)\) and
    \(\widehat{\mathcal J}_n^1(t+1)\) using
    \eqref{eq:rvou_linear_heads}--\eqref{eq:predicted_reliability_cost}
    \State \(s_n(t)
    \gets
    \widehat{\mathcal J}_n^0(t+1)
    -
    \widehat{\mathcal J}_n^1(t+1)\)
\EndFor
\State \(S_t\gets\operatorname{TopK}_{+}(s(t))\)
\State Pull \(\{Y_n(t):n\in S_t\}\) and set
\(u_n(t)=\mathbf 1\{n\in S_t\}\)
\State For all \(m\), form
\(\mathcal B_t^{(m)}
\gets\{(q_n(t),Y_n(t)-\tilde X_n^{(m)}(t)):n\in S_t\}\)
\State Update all DT states using
\eqref{eq:base_dt_recursion_joint}--\eqref{eq:corrected_dt_recursion_joint}
with \(u(t)\) and \(\mathcal M_t\)
\State For all \(m\), update
\(\phi_{t+1}^{(m)}
\gets \mathcal A_{\mathrm{RLS}}^{(m)}
(\phi_t^{(m)},\mathcal B_t^{(m)})\)
\State \(\mathcal M_{t+1}\gets\{\phi_{t+1}^{(m)}\}_{m=1}^{M}\)
\end{algorithmic}
\end{algorithm}
\vspace{-0.3mm}
For fixed feature dimension, score computation is \(O(N)\), and selecting
\(S_t\) by sorting costs \(O(N\log N)\). The correction state
\(\mathcal M_t\) is used during slot \(t\), while the updated state
\(\mathcal M_{t+1}\) is used from the next slot onward.
\vspace{-2.45mm}
\section{Numerical Evaluation}
\label{sec:numerical}
\vspace{-0.7mm}
We evaluate R-VoU using recorded measurements from a battery production
process, which serves as a representative example of IIoT process monitoring.
In such settings, distributed sensing tracks the evolution of a physical
process, while communication constraints limit how many observations can be
collected at each slot. Each sampled location is represented by a scalar state,
so \(d=1\). The monitored spatial profile is modeled by \(N=10\) uniformly
sampled positions, where \(\xi_n\) is the normalized distance of position
\(n\) to the nearest profile boundary. Boundary positions receive larger
priority weights. A stable segment \(X_1\) is used to pre-train the
device-level ensemble predictors, while a later segment \(X_2\) is used for
evaluation under concept drift. Unless otherwise stated, each device uses
\(M=5\) base predictors and the shared online correction module.

We compare R-VoU with weighted Age of Information (wAoI), EDI-VoU,
and round-robin (RR) under budgets \(K\in\{1,2,3,4\}\). EDI-VoU is the
uncertainty-only ablation obtained by setting \(\alpha=1\), while R-VoU uses
\(\alpha=0.3\). To enable a comparison with Age of Incorrect Information (AoII), which
combines information mismatch with the time spent in an incorrect state~\cite{maataouk2020aoii}, we
include AoII\(^{\dagger}\) as a noncausal reference. In our DT setting, such
mismatch is the instantaneous error between the receiver-side DT estimate and
the recorded reference state. AoII\(^{\dagger}\) prioritizes devices using
this mismatch together with the time elapsed since the last successful pull.
Since the recorded reference trajectory is not available to the BS before
scheduling during deployment, AoII\(^{\dagger}\) is not a causal runtime
scheduler.

The predictive heads are trained on a prefix of \(40\%\) of \(X_2\) and evaluated on a
disjoint suffix. The feature vector $z_n(t)$ contains only pre-decision BS-observable features: EDI, edge-position weights, past residual summaries, covariance uncertainty, and local coupling terms. We report measured EDI \(J_I\), receiver-side DT error
\(J_e\), and the composite reliability cost \(J_{\mathcal J}\).

\begin{figure}[t]
    \centering
    \includegraphics[width=0.99\columnwidth]{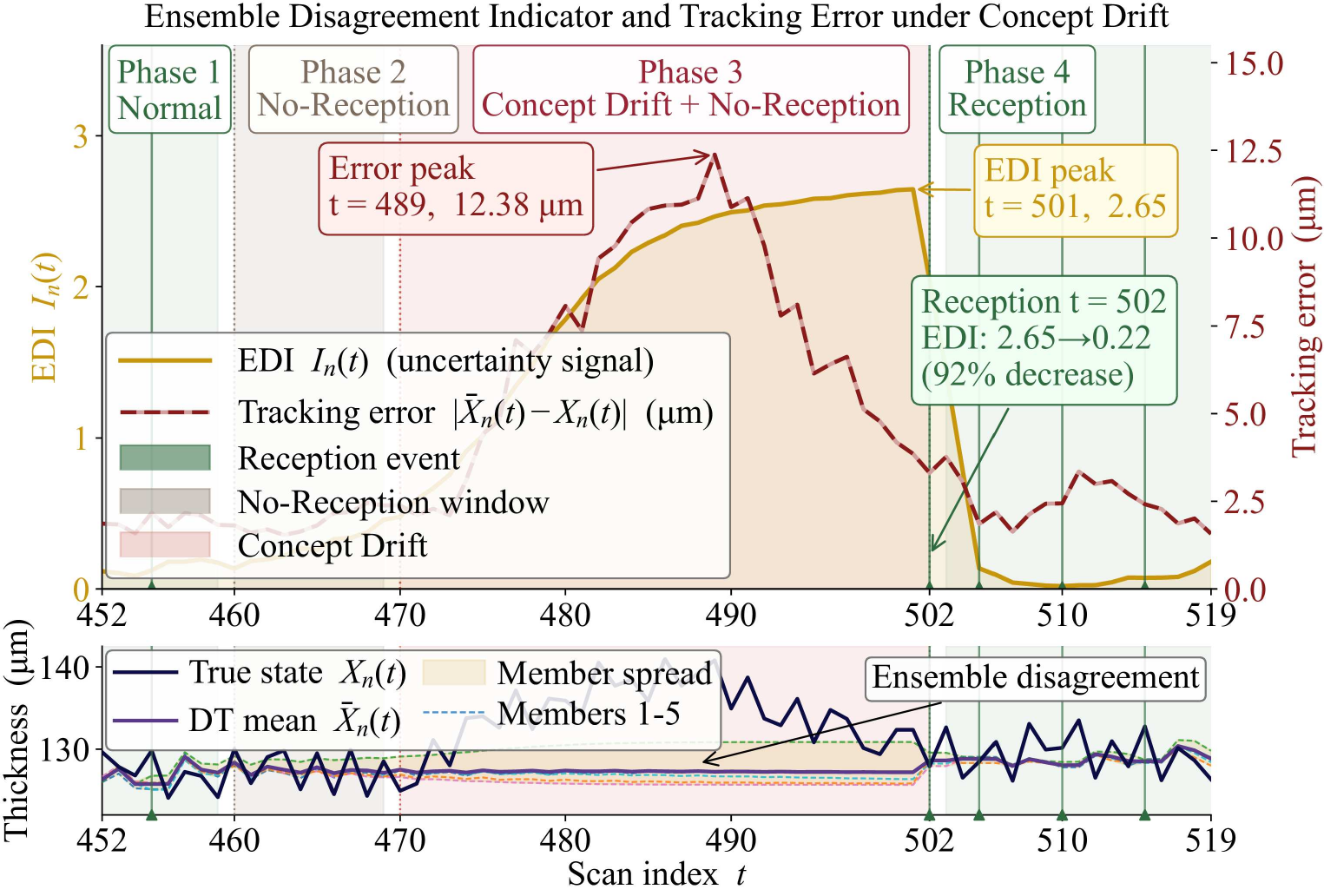}
    \caption{Representative position under concept drift. During a passive
    interval, EDI increases with DT error; after a successful reception,
    EDI drops sharply.}
    \label{fig:aoint_mechanism}
\end{figure}

\vspace{-1.9mm}
\subsection{EDI as a DT-Reliability Signal under Concept Drift}

We first evaluate whether EDI is associated with receiver-side DT error
using passive slots, i.e., slots without a successful pull for the considered
position. Pooling \(2805\) passive slots gives a positive Spearman correlation
\(\rho=0.3265\). The pooled mean DT error increases
from \(0.710\,\mu\mathrm{m}\) in the lowest EDI quintile to
\(3.540\,\mu\mathrm{m}\) in the highest, corresponding to a \(4.99\times\)
increase. Fig.~\ref{fig:aoint_mechanism} further illustrates that EDI rises
during a no-reception drift interval and drops after a successful update. These
results motivate combining EDI with predicted DT error in R-VoU.
\vspace{-1.6mm}
\subsection{Reliability-Aware Pull Scheduling}
\vspace{-0.3mm}
Table~\ref{tab:exp2_rollout} reports the scheduling results. Among the
causal schedulers, R-VoU achieves the lowest receiver-side DT error \(J_e\)
and the lowest composite reliability cost \(J_{\mathcal J}\) across all
evaluated budgets. Its gains over the strongest causal baseline reach up to
\(5.5\%\) in \(J_e\) and \(2.1\%\) in \(J_{\mathcal J}\). The policy that
minimizes measured EDI \(J_I\) varies across budgets, which further shows
that EDI alone is not equivalent to receiver-side DT reliability.

The noncausal AoII\(^{\dagger}\) benchmark achieves the lowest \(J_e\) by using recorded DT error computed before scheduling, which is unavailable to the BS at runtime. In contrast, R-VoU uses only BS-observable
information, remains close to AoII\(^{\dagger}\) at larger budgets, and
achieves the lowest \(J_{\mathcal J}\) at \(K=4\).

The online correction module further improves receiver-side DT reliability.
At \(K=3\), enabling RLS correction reduces \(J_e\) and \(J_{\mathcal J}\) by \(55.9\%\)  and \(28.4\%\), respectively, relative to the same R-VoU policy without correction.
\vspace{-0.2mm}
\begin{table}[!t]
\caption{Scheduling performance for \(N=10\) across budgets
\(K\in\{1,2,3,4\}\).}
\label{tab:exp2_rollout}
\vspace{-0.6mm}
\centering
\tiny
\renewcommand{\arraystretch}{0.99}
\setlength{\tabcolsep}{1.20pt}
\setlength{\arrayrulewidth}{0.6pt}
\resizebox{0.99\columnwidth}{!}{%
\begin{tabular}{@{}l|rrrrr|rrrrr@{}}
\hline
\multicolumn{1}{c|}{} 
& \multicolumn{5}{c|}{$K=1$}
& \multicolumn{5}{c}{$K=2$} \\
\cline{2-11}
Metric
& wAoI & RR & EDI-VoU & R-VoU & AoII\(^{\dagger}\)
& wAoI & RR & EDI-VoU & R-VoU & AoII\(^{\dagger}\) \\
\hline
$J_I$ &
\textbf{2.64} & 4.10 & 3.44 & 3.72 & 3.48 &
\textbf{3.12} & 3.85 & 5.93 & 4.00 & 4.11 \\
$J_e$ &
82.93 & 100.88 & 85.62 & \textbf{78.38} & 76.25 &
75.89 & 84.25 & 117.78 & \textbf{73.62} & 67.67 \\
$J_{\mathcal J}$ &
101.43 & 128.15 & 108.68 & \textbf{100.70} & 97.69 &
94.86 & 88.37 & 125.36 & \textbf{86.52} & 74.35 \\
\hline
\multicolumn{1}{c|}{} 
& \multicolumn{5}{c|}{$K=3$}
& \multicolumn{5}{c}{$K=4$} \\
\cline{2-11}
Metric
& wAoI & RR & EDI-VoU & R-VoU & AoII\(^{\dagger}\)
& wAoI & RR & EDI-VoU & R-VoU & AoII\(^{\dagger}\) \\
\hline
$J_I$ &
3.57 & 8.28 & 10.17 & \textbf{2.75} & 4.20 &
4.36 & 6.50 & \textbf{2.05} & 3.06 & 3.29 \\
$J_e$ &
69.13 & 84.81 & 99.19 & \textbf{66.14} & 64.98 &
66.18 & 76.96 & 65.67 & \textbf{64.13} & 63.74 \\
$J_{\mathcal J}$ &
110.28 & 152.60 & 180.42 & \textbf{108.69} & 108.48 &
109.22 & 133.83 & 103.87 & \textbf{101.89} & 102.26 \\
\hline
\end{tabular}%
}
\vspace{0.2mm}

\noindent
\begin{minipage}{0.99\columnwidth}
\raggedright\footnotesize
\hspace*{-0.4mm}\(^{\dagger}\)
Noncausal benchmark using the recorded source--receiver mismatch.
\end{minipage}
\vspace{-0.5mm}
\end{table}
\vspace{-0.7mm}
\section{Conclusion}

This paper studied reliability-aware pull scheduling for receiver-side DT
maintenance under limited uplink resources and concept drift. We defined EDI
as a BS-observable uncertainty signal from ensemble disagreement and
distinguished it from receiver-side DT error, which is latent at runtime. This
distinction motivated R-VoU, a value-of-update scheduler that predicts
next-slot EDI and DT error risk and pulls devices according to the predicted
reduction in composite reliability cost. Experiments on recorded process
manufacturing data showed that EDI is positively associated with DT error,
but is not sufficient as a standalone reliability objective. Across the
evaluated budgets, R-VoU achieved the lowest composite reliability cost and
receiver-side DT error among the causal schedulers. The online correction
module further improved reliability through residual-based adaptation. Future
work will study richer DT error risk predictors and longer-horizon scheduling
policies.
\vspace{-2.950mm}
\bibliographystyle{IEEEtran}
\bibliography{references}

\end{document}